\documentclass[a4paper]{svmult}

\usepackage[english]{babel}
\usepackage[utf8]{inputenc}
\usepackage[T1]{fontenc}
\usepackage{lmodern}
\usepackage{microtype}

\usepackage{amsmath}
\usepackage{amssymb}

\let\E\relax

\usepackage[linesnumbered,noend]{algorithm2e}
\usepackage[small,bold]{complexity}
\usepackage{psystems}

\newcommand{\set}[1]{\{#1\}}
\newcommand{\N}{\mathbb{N}}

\newcommand{\yes}{\textsf{yes}}
\newcommand{\no}{\textsf{no}}

\smartqed

\title*{Characterizing $\PSPACE$ with Shallow Non-Confluent P~Systems}

\author{Alberto Leporati \and Luca Manzoni \and Giancarlo Mauri
  \and Antonio E. Porreca \and Claudio Zandron}

\authorrunning{Alberto Leporati et al.}

\institute{Dipartimento di informatica, Sistemistica e Comunicazione\\
  Università degli Studi di Milano-Bicocca, \\
  Viale Sarca 336, 20126, Milan, Italy\\
  \texttt{\{leporati, luca.manzoni, mauri, porreca, zandron\}@disco.unimib.it}
}

\begin{document}

\maketitle

\begin{abstract}
  In P~systems with active membranes, the question of understanding the power of non-confluence within a polynomial time bound is still an open problem. It is known that, for shallow P~systems, that is, with only one level of nesting, non-confluence allows them to solve conjecturally harder problems than confluent P~systems, thus reaching $\PSPACE$. Here we show that $\PSPACE$ is not only a bound, but actually an exact characterization. Therefore, the power endowed by non-confluence to shallow P~systems is equal to the power gained by \emph{confluent} P~systems when non-elementary membrane division and polynomial depth are allowed, thus suggesting a connection between the roles of non-confluence and nesting depth.
\end{abstract}

\section{Introduction}
\label{sec:introduction}

While families of \emph{confluent} recognizer P~systems with active membranes with charges are known to characterize the complexity class $\PSPACE$ when working in polynomial time~\cite{Sosik2003a,Sosik2007a}, their computational power when the nesting level is constrained to one (i.e., only one level of membranes inside the outermost membrane, usually called \emph{shallow} P~systems) is reduced to the class $\P^{\boldsymbol\#\P}$, which is conjecturally smaller~\cite{Leporati2014d}. While confluent P~systems can make use of nondeterminism, they are constrained in returning the same result for all computations starting from the same initial configuration. However, by accepting when at least one computation accepts, like nondeterministic Turing Machines (TM) traditionally do, P~systems can make use of the entire power of nondeterminism: uniform families of \emph{non-confluent} recognizer P~systems with active membranes with charges can solve $\PSPACE$-complete problems even in the shallow case and even when send-in rules are disallowed (i.e., for monodirectional systems)~\cite{Leporati2016a}. Here we show that, in fact, $\PSPACE$ is a characterization of this kind of shallow non-confluent P~systems when they work in polynomial time. This result shows that the complex relation between computational power, nesting depth, and monodirectionality present for confluent P~systems is absent in the non-confluent case. In particular, in the confluent case, systems with no nesting characterize $\P$~\cite{Zandron2001a} whereas, additional nesting gives additional power~\cite{Leporati2014e} until reaching $\PSPACE$ when unlimited nesting is allowed~\cite{Sosik2003a,Sosik2007a}. In the monodirectional case even unlimited nesting cannot escape $\P^\NP$, which is conjecturally smaller~\cite{Leporati2016b}. Non-confluent systems, on the other hand, characterize $\NP$ when there are no internal membranes~\cite{Porreca2010a}, and immediately gain the full power of $\PSPACE$ with only one level of nesting. Furthermore, at least for shallow systems, this provides an exact characterization. It is therefore natural to ask what is the relation between the mechanisms that empower confluent P~systems and the full power of non-confluence. Are the former ones only a way to simulate the latter?

\section{Basic Notions}
\label{sec:basic-notions}

For an introduction to membrane computing and the related notions of formal language theory, we refer the reader to \emph{The Oxford Handbook of Membrane Computing}~\cite{Paun2010a}. Here we recall the formal definition of P~systems with active membranes using only elementary division rules.

\begin{definition}
  A \emph{P~system with active membranes with elementary division rules} of initial degree~$d \ge 1$ is a tuple
  \begin{align*}
    \Pi = (\Gamma, \Lambda, \mu, w_{h_1}, \ldots, w_{h_d}, R)
  \end{align*}
  where:
  \begin{itemize}
    \item $\Gamma$ is an alphabet, i.e., a finite non-empty set of symbols, usually called \emph{objects};
    \item $\Lambda$ is a finite set of labels for the membranes;
    \item $\mu$ is a membrane structure (i.e., a rooted \emph{unordered} tree, usually represented by nested brackets) consisting of~$d$ membranes labelled by elements of~$\Lambda$ in a one-to-one way;
    \item $w_{h_1}, \ldots, w_{h_d}$, with~$h_1, \ldots, h_d \in \Lambda$, are multisets (finite sets whose elements have a multiplicity) of objects in~$\Gamma$, describing the initial contents of the~$d$ regions of~$\mu$;
    \item $R$ is a finite set of rules.
  \end{itemize}
\end{definition}

\noindent
Each membrane possesses, besides its label and position in~$\mu$, another attribute called \emph{electrical charge}, which can be either neutral~($0$), positive~($+$) or negative~($-$) and is always neutral before the beginning of the computation.

The rules in~$R$ are of the following types:
\begin{enumerate}
  \item[(a)] \emph{Object evolution rules}, of the form~$\pevolve{h}{\alpha}{a}{w}$

  They can be applied inside a membrane labelled by~$h$, having charge~$\alpha$ and containing an occurrence of the object~$a$; the object~$a$ is rewritten into the multiset~$w$ (i.e.,~$a$ is removed from the multiset in~$h$ and replaced by the objects in~$w$).

  \item[(b)] \emph{Send-in communication rules}, of the form $a \, [\;]_h^\alpha \to [b]_h^\beta$

  They can be applied to a membrane labelled by $h$, having charge $\alpha$ and such that the external region contains an occurrence of the object $a$; the object $a$ is sent into $h$ becoming $b$ and, simultaneously, the charge of $h$ is changed to~$\beta$.

  \item[(c)] \emph{Send-out communication rules}, of the form~$\psendout{h}{\alpha}{a}{\beta}{b}$

  They can be applied to a membrane labelled by~$h$, having charge~$\alpha$ and containing an occurrence of the object~$a$; the object~$a$ is sent out from~$h$ to the outside region becoming~$b$ and, simultaneously, the charge of~$h$ becomes~$\beta$.

  \item[(e)] \emph{Elementary division rules}, of the form~$\pdivide{h}{\alpha}{a}{\beta}{b}{\gamma}{c}$

  They can be applied to a membrane labelled by~$h$, having charge~$\alpha$, containing an occurrence of the object~$a$ but having no other membrane inside (an \emph{elementary membrane}); the membrane is divided into two membranes having label~$h$ and charges~$\beta$ and~$\gamma$; the object~$a$ is replaced, respectively, by~$b$ and~$c$, while the other objects of the multiset are replicated in both membranes.
\end{enumerate}

The instantaneous \emph{configuration} of a membrane of label $h$ consists of its charge~$\alpha$ and the multiset~$w$ of objects it contains at a given time. It is denoted by~$\membrane{h}{\alpha}{w}$. The \emph{(full) configuration}~$\confC$ of a P~system~$\Pi$ at a given time is a rooted, unordered tree. The root is a node corresponding to the external environment of~$\Pi$, and has a single subtree corresponding to the current membrane structure of~$\Pi$. Furthermore, the root is labelled by the multiset located in the environment, and the remaining nodes by the configurations~$\membrane{h}{\alpha}{w}$ of the corresponding membranes. In the \emph{initial configuration} of $\Pi$, the configurations of the membranes are $\membrane{h_1}{0}{w_{h_1}}, \ldots, \membrane{h_d}{0}{w_{h_d}}$.

A P~system is \emph{shallow} if it contains at most one level of membranes inside the outermost membrane. This means that all the membranes contained in the outermost membrane are elementary, i.e., they contain no other nested membrane.

A computation step changes the current configuration according to the following set of principles:
\begin{itemize}
  \item Each object and membrane can be subject to at most one rule per  step, except for object evolution rules: inside each membrane, several evolution rules can be applied simultaneously.
  \item The application of rules is \emph{maximally parallel}: each object appearing on the left-hand side of evolution, communication, or division rules must be subject to exactly one of them (unless the current charge of the membrane prohibits it). Analogously, each membrane can only be subject to one communication or division rule (types (b)--(e)) per computation step; these rules will be called \emph{blocking rules} in the rest of the paper. In other words, the only objects and membranes that do not evolve are those associated with no rule, or only to rules that are not applicable due to the electrical charges.
  \item When several conflicting rules can be applied at the same time, a nondeterministic choice is performed; this implies that, in general, multiple possible configurations can be reached after a computation step.
  \item In each computation step, all the chosen rules are applied simultaneously (in an atomic way). However, in order to clarify the operational semantics, each computation step is conventionally described as a sequence of micro-steps whereby each membrane evolves only after their internal configuration (including, recursively, the configurations of the membrane substructures it contains) has been updated. In particular, before a membrane division occurs, all chosen object evolution rules must be applied inside it; this way, the objects that are duplicated during the division are already the final ones.
  \item The outermost membrane cannot be divided, and any object sent out from it cannot re-enter the system again.
\end{itemize}
A \emph{halting computation} of the P~system~$\Pi$ is a finite sequence~$\compC = (\confC_0, \ldots, \confC_k)$ of configurations, where~$\confC_0$ is the initial configuration, every~$\confC_{i+1}$ is reachable from~$\confC_i$ via a single computation step, and no rules of~$\Pi$ are applicable in~$\confC_k$. A \emph{non-halting} computation~$\compC = (\confC_i : i \in \N )$ consists of infinitely many configurations, again starting from the initial one and generated by successive computation steps, where the applicable rules are never exhausted.

P~systems can be used as language \emph{recognisers} by employing two distinguished objects~$\yes$ and~$\no$: in this case we assume that all computations are halting, and that either one copy of object~$\yes$ or one of object~$\no$ is sent out from the outermost membrane, and only in the last computation step, in order to signal acceptance or rejection, respectively. If all computations starting from the same initial configuration are accepting, or all are rejecting, the P~system is said to be \emph{confluent}. In this paper we deal, however, with \emph{non-confluent} P~systems, where multiple computations can have different results and the overall result is established as for nondeterministic TM: it is acceptance iff an accepting computation exists~\cite{Perez2003a}.

In order to solve decision problems (or, equivalently, decide languages) over an alphabet $\Sigma$, we use \emph{families} of recogniser P~systems $\familyPi = \{ \Pi_x : x \in \Sigma^\star \}$. Each input~$x$ is associated with a P~system~$\Pi_x$ deciding the membership of~$x$ in a language $L \subseteq \Sigma^\star$ by accepting or rejecting. The mapping~$x \mapsto \Pi_x$ must be efficiently computable for inputs of any length, as discussed in detail in~\cite{Murphy2011a}.

\begin{definition}
  \label{def:uniform}
  A family of P~systems~$\familyPi = \{ \Pi_x : x \in \Sigma^\star \}$ is \emph{(polynomial-time) uniform} if the mapping~$x \mapsto \Pi_x$ can be computed by two polynomial-time deterministic Turing machines~$E$ and~$F$ as follows:
  \begin{itemize}
    \item $F(1^n) = \Pi_n$, where~$n$ is the length of the input~$x$ and~$\Pi_n$ is a common P~system for all inputs of length~$n$, with a distinguished input membrane.
    \item $E(x) = w_x$, where~$w_x$ is a multiset encoding the specific input~$x$.
    \item Finally,~$\Pi_x$ is simply~$\Pi_n$ with~$w_x$ added to its input membrane.
  \end{itemize}
  The family~$\familyPi$ is said to be (polynomial-time) semi-uniform if there exists a single deterministic polynomial-time Turing machine~$H$ such that~$H(x) = \Pi_x$ for each~$x \in \Sigma^\star$.
\end{definition}

Any explicit encoding of~$\Pi_x$ is allowed as output of the construction, as long as the number of membranes and objects represented by it does not exceed the length of the whole description, and the rules are listed one by one. This restriction is enforced in order to mimic a (hypothetical) realistic process of construction of the P~systems, where membranes and objects are presumably placed in a constant amount during each construction step, and require actual physical space proportional to their number; see also~\cite{Murphy2011a} for further details on the encoding of P~systems.

In the following, we denote the class of problems solvable by polynomial-time uniform or semi-uniform families of non-confluent shallow P~systems with active membranes with charges by $\NPMC_{\pAM[depth\mbox{-}1,-d,-ne]}^{[\star]}$, where $[\star]$ denotes optional semi-uniformity. If no restriction on the depth of the membrane structure is present, but both non-elementary division and dissolution rules are forbidden, then the corresponding class of problems is denoted by $\NPMC_{\pAM[-d,-ne]}^{[\star]}$.

\section{Nondeterministic Simulation with Oracles}

Let $\familyPi$ be a semi-uniform family of non-confluent shallow recognizer P~systems with active membranes with charges, and let $H$ be the TM of the semi-uniformity condition of $\familyPi$. We are going to define a machine $M$ working in polynomial \emph{space} such that on input $H$ and $x$ Turing machine $M$ accepts iff the P~system $H(x) = \Pi_x$ of $\familyPi$ accepts in polynomial \emph{time}. Notice that a single machine $M$ suffices for all families of P~systems. The machine associated with a specific family $\familyPi$ of P~systems can be obtained by ``hard-coding'' the input $H$ to $M$.

First of all, on input $H$ and $x$, machine $M$ simulates machine $H$ with $x$ as input to obtain a polynomial-size description of $\Pi_x$. To simplify the description of the procedure used by machine $M$ to simulate $\Pi_x$, we will assume $M$ to work as a nondeterministic polynomial-time TM with access to an oracle for a problem in $\NPSPACE=\PSPACE$. As the following result shows, both this nondeterministic behaviour and the oracle queries can still all be simulated using a polynomial-space deterministic TM.

\begin{proposition}
  \label{thm:still-pspace}
  $\NP^\NPSPACE = \PSPACE$.
\end{proposition}
\begin{proof}
  Clearly $\NP^\NPSPACE \supseteq \PSPACE$, hence only the opposite inclusion needs to be proved. Let $N$ be a polynomial-time nondeterministic TM with access to an oracle for a language $L \in \NPSPACE$. Let $D$ be a deterministic polynomial space TM built in the following way:
  \begin{itemize}
    \item $D$ simulates $N$ until a query is performed. This simulation, including the nondeterministic choices of $N$, can be performed in polynomial space by $D$, since $\NP \subseteq \PSPACE$.
    \item Since $L \in \NPSPACE$ and $\NPSPACE = \PSPACE$, there exists a deterministic polynomial space TM deciding $L$ that can be simulated by $D$ to answer any query performed by $N$ while still using only a polynomial amount of space. Once a query has been answered, $D$ can resume the simulation of $N$.
  \end{itemize}
  Since $D$ can faithfully simulate $N$ and its oracle queries, $D$ can recognize the same language as $N$, thus showing that $\NP^\NPSPACE \subseteq \PSPACE$, as desired.
  \qed
\end{proof}

We can now describe how the simulation of $\Pi_x$ is carried on by $M$. In the following, we assume that the size of the input $x$ is $n$, and that each computation of $\Pi_x$ requires at most $T$ time steps before halting and producing a result. By hypothesis $T$ is polynomial with respect to $n$.

\subsection{Simulation of the Outermost Membrane}
\label{sec:nondeterministic-sim}

The main idea of this construction is to simulate the evolution of the outermost membrane directly by means of a nondeterministic polynomial-time TM. All interactions with the internal membranes are performed via nondeterministic guesses. That is, for each communication rule and for each time step, the number of rules that are applied between the outermost and the inner membranes is guessed in a nondeterministic way. If $\yes$ has been sent out by the simulation of the outermost membrane, an oracle query is performed to check whether all performed interactions with the inner membranes were \emph{consistent} with this information, that is, if a computation of the inner membranes able to perform the guessed interactions actually exists. If the query returns a positive answer, then a computation of the entire system actually producing $\yes$ exists. In any other case, the simulating machine rejects (since either an invalid simulation of the outermost membrane -- and of the P~system -- was produced, or the simulation itself was correct but the simulated computation was a rejecting one).

\newcommand{\iTable}{\mathcal{T}}
\newcommand{\uTable}{\mathcal{U}}
\newcommand{\maxT}{\mathsf{K}}
To perform this construction we build a table $\iTable$ indexed by pairs of the form $(r, t)$, where $r \in R$ is either a send-in rule from the outermost membrane to one of the internal membranes or a send-out rule from one of the internal membranes to the outermost membrane, and $t \in \set{0, \ldots, T-1}$ is a time step. The entry $\iTable(r, t)$ represent the number of times rule $r$ has been applied at the time step $t$. It is important to notice that table $\iTable$ can be stored using a polynomial amount of space. In fact, the number of entries is limited by the size of $R$ (which, by uniformity condition, is polynomial in the input size $n$), and by the number $T$ of time steps needed for the P~system to halt. We only need to prove that each entry $\iTable(r,t)$ can be stored in a polynomial amount of space.

Let $m \in \N$ be number of internal membranes in the initial configuration of $\Pi_x$. By the semantics of the rules of P~systems, the number of objects sent in to internal membranes or sent out from them after $t$ time steps cannot be greater than $m \times 2^t$, where the second multiplicative factor is the maximum number of membranes per label that can be obtained by membrane division in $t$ time steps. Since this value is exponential in $t$, it can be represented by a polynomial number of bits with respect to $t \le T$. Thus, each entry of $\iTable$ requires at most a polynomial amount of space with respect to $n$. We denote the maximum value attainable by an entry of $\iTable$ by $\maxT$.

Apart from keeping track of the communication rules applied between the outermost and the internal membranes, we also need to assure that all rules are applied in a maximally parallel way. To do so, we define another table $\uTable$ indexed by pairs of the form $(a, t)$ where $a \in \Gamma$ is an object type and $t \in \set{0, \ldots, T-1}$ is, as before, a time step. The entry $\uTable(a, t)$ represents the number of objects of type $a$ in the outermost membrane that had no rule applied to them at time $t$. Table $\uTable$ can, too, be stored in a polynomial amount of space.

The simulation procedure of the outermost membrane is detailed as Algorithm~\ref{alg:outermost}. There, label $h$ always indicates the outermost membrane and the label $k$ an internal membrane label, while $|w|_a$ denotes the number of instances of the object $a$ inside the multiset $w$. The \emph{applicability} of a rule refers, in the algorithm, to the fact that the indicated membrane must have the correct charge $\alpha$ and, if the rule is \emph{blocking}, that the membrane has not already been used by another blocking rule in the same time step. For example, the condition on line~\ref{alg1:to-env-applicable} of Algorithm~\ref{alg:outermost} is never verified once another send-out rule has been simulated in a previous iteration of the loop for the current time step.

\begin{algorithm}
  \SetKw{KwGuess}{guess}
  \SetKw{KwQuery}{query}

  $w \leftarrow$ initial multiset of the outermost membrane\;  \label{alg1:init}
  $\mathsf{env} \leftarrow \varnothing$\;
  $\mathsf{charge} \leftarrow 0$\;  \label{alg1:init-end} 
  \For{$t \leftarrow 0$ \KwTo $T-1$}{ \label{alg1:for-cycle-begin}
    \For{all applicable $r = \psendin{k}{\alpha}{a}{\beta}{b}$}{ \label{alg1:send-in-applicable}
      $\iTable(r, t) \leftarrow $ \KwGuess($0$, $\min(|w|_a, \maxT)$)\; \label{alg1:guess-send-in-exec}
      mark $\iTable(r, t)$ instances of $a$ for removal from $w$\; \label{alg1:mark-send-in}
    } \label{alg1:send-in-applicable-end}
    \For{$r = \psendout{k}{\alpha}{a}{\beta}{b}$}{ \label{alg1:send-out-applicable}
      $\iTable(r, t) \leftarrow$ \KwGuess($0$, $\maxT$)\; \label{alg1:guess-send-out-exec}
      mark $\iTable(r, t)$ instances of $b$ for insertion in $w$\; \label{alg1:mark-send-out}
    } \label{alg1:send-out-applicable-end}
    \For{all applicable $r = \pevolve{h}{\alpha}{a}{u}$}{ \label{alg1:evolve-applicable}
      $m \leftarrow$ \KwGuess($0$, $|w|_a$)\; \label{alg1:guess-evolve}
      mark $m$ copies of $u$ for addition to $w$ and $m$ copies of $a$ for removal\; \label{alg1:mark-evolve}
    } \label{alg1:evolve-applicable-end}
    \For{all applicable $r = \psendout{h}{\alpha}{a}{\beta}{b}$}{ \label{alg1:to-env-applicable}
      $m \leftarrow$ \KwGuess($0$, $1$)\; \label{alg1:guess-to-env}
      \If{$m = 1$}{ \label{alg1:if-apply}
        mark one copy of $a$ for removal from $w$\; \label{alg1:obj-removal}
        mark one copy of $b$ for addition in $\mathsf{env}$\; \label{alg1:obj-add}
        mark $\mathsf{charge}$ to be changed from $\alpha$ to $\beta$\; \label{alg1:change-charge}
      } \label{alg1:if-apply-end}
    } \label{alg1:to-env-applicable-end}
    \For{$a \in \Gamma$}{ \label{alg1:collect-unused}
      $\uTable(a, t) \leftarrow$ number of instances of $a$ in $w$ not marked\; \label{alg1:update-utable}
    } \label{alg1:collect-unused-end}
    Apply modifications to $w$, $\mathsf{env}$, and $\mathsf{charge}$ according to the markings\; \label{alg1:modify-state}
    \If{rule application was not maximally parallel}{ \label{alg1:not-maximally-parallel}
      reject\; \label{alg1:reject-no-maximaly parallel}
    } \label{alg1:not-maximally-parallel-end}
    \If{$\yes$ or $\no$ has been sent out in the environment}{ \label{alg1:yes-or-no}
      \eIf{\KwQuery$(\iTable, \uTable, t)$ answer is positive and no further rules are applicable in the next time step}{ \label{alg1:query}
        accept or reject accordingly\; \label{alg1:accordingly}
      }{ \label{alg1:else-we-were-wrong}
        reject\; \label{alg1:reject-wrong-guess}
      } \label{alg1:query-end}
    } \label{alg1:yes-or-no-end}
  } \label{alg1:for-cycle-end}
  reject\; \label{alg1:reject-too-much-time}
  \caption{\label{alg:outermost}The nondeterministic algorithm that performs the simulation of the outermost membrane of $\Pi_x$.}
\end{algorithm}

Lines~\ref{alg1:init}--\ref{alg1:init-end} perform the initialization of the algorithm, setting the initial content and charge of the outermost membrane and declaring the environment initially empty. The main simulation loop is performed in lines \ref{alg1:for-cycle-begin}--\ref{alg1:for-cycle-end}. Since the maximum number of time steps needed for $\Pi_x$ to produce a result is $T$, the simulation loop is repeated at most $T$ times. If the loop ends without having produced either $\yes$ on $\no$ in the environment while simultaneously halting, the simulation performed did not correspond to any actual computation of $\Pi_x$, thus a negative answer must be produced (line~\ref{alg1:reject-too-much-time}).

Lines~\ref{alg1:send-in-applicable}--\ref{alg1:send-in-applicable-end} deal with the send-in rules from the outermost membrane to the inner membranes. Since the number of internal membranes where the rule $r$ can be applied is not known, the number is nondeterministically chosen and is bounded by the maximum number of inner membranes and the number of objects of type $a$ in the outermost membrane (line~\ref{alg1:guess-send-in-exec}). The guessed number of internal membranes saved in table $\iTable$ and the effect of the rules on the multiset $w$ is scheduled for application (line~\ref{alg1:mark-send-in}). Notice that, since the state of the internal membranes is not stored, this amounts to the removal of $\iTable(r, t)$ instances of objects of type $a$ from $w$.

Lines~\ref{alg1:send-out-applicable}--\ref{alg1:send-out-applicable-end} deal with send-out rules from the internal membranes to the outermost membrane. As before, since the configuration and number of the internal membranes is not known, the number of times this rule is applied is nondeterministically guessed (line~\ref{alg1:guess-send-out-exec}), saved in table $\iTable$, and the appearance of the corresponding objects of type $b$ in $w$ is scheduled (line~\ref{alg1:mark-send-out}).

Lines~\ref{alg1:evolve-applicable}--\ref{alg1:evolve-applicable-end} perform the simulation of the evolution rules inside the outermost membrane. Since the simulated system is non-confluent, the actual number of applications of each rule is guessed (line~\ref{alg1:guess-evolve}) before the actual effect of the rule applications are scheduled (line~\ref{alg1:mark-evolve}).

Lines~\ref{alg1:to-env-applicable}--\ref{alg1:to-env-applicable-end} deal with the application of send-out rules from the outermost membrane to the environment. First of all, a nondeterministic guess is performed to decide whether the rule is actually applied (line~\ref{alg1:guess-to-env}). If so, then the actual effects of the rules are scheduled for application (lines~\ref{alg1:if-apply}--\ref{alg1:if-apply-end}).

The table $\uTable$ is then updated to memorize the number of objects that were not subjected to any rule (lines~\ref{alg1:collect-unused}--\ref{alg1:collect-unused-end}). This will be used during the query process to ensure that the send-in rules from the outermost membrane to the internal membranes were actually applied in a maximally parallel way.

All the scheduled modifications to the content and charge of the outermost membrane and to the environment are now executed (line~\ref{alg1:modify-state}). If there are irreconcilable problems in the maximally parallel application of the rules then a rejection is performed (lines~\ref{alg1:not-maximally-parallel}--\ref{alg1:not-maximally-parallel-end}). This happens when there were objects in the outermost membrane that were not selected to be sent-in into the internal membranes (this can be checked by looking at table $\uTable$), nor were they subject to applicable send-out or evolution rules.

Finally, if either $\yes$ or $\no$ appears in the environment (lines~\ref{alg1:yes-or-no}--\ref{alg1:yes-or-no-end}) then it is necessary to check whenever the guesses performed for the interaction with the internal membranes were accurate and no further rules are applicable in the next time step in the outermost membrane (lines~\ref{alg1:query}--\ref{alg1:query-end}). If the answer to the query is positive and no further rules were actually applicable, then the simulation can either accept or reject accordingly (line~\ref{alg1:accordingly}). Otherwise, the simulation performed did not correspond to any actual computation of $\Pi_x$ and we must reject (line~\ref{alg1:reject-wrong-guess}).

Algorithm~\ref{alg:outermost} can be executed in polynomial time by a nondeterministic TM with access to an oracle to perform the query procedure. In fact, both the outer loop and the inner loops are executed only a polynomial amount of times (either bounded by the time needed for $\Pi_x$ to halt or by the number or rules in the system). All other operations, including checking the applicability of rules, can be performed in polynomial time given an efficient description of the configuration of the outermost membrane (in which the number of objects is stored in binary). Furthermore, all nondeterministic guesses are of a polynomial amount of bits.

\subsection{Simulation of the Oracle}
\label{sec:queries}

The query that is simulated by means of a nondeterministic machine working in polynomial space is the following one:
\begin{quote}
  Is there an halting computation of length $t$ of the internal membranes consistent with the rule applications guessed?
\end{quote}
To be able to answer this query in nondeterministic polynomial space the main idea is to simulate each membrane sequentially and keep track of the communication rules that are applied while comparing them with the ones guessed by the simulation of the outermost membrane. If division is applied then only the simulation of one of the dividing membranes is immediately carried out (as performing them all at the same time might require exponential -- instead of polynomial -- space) while the other membrane is pushed into a stack, thus performing a \emph{depth-first simulation} of the membrane hierarchy. This ensures that a polynomial amount of space suffices: it the space needed to simulate one membrane, plus a stack in which the number of elements is at most $T$, one for each time step. This algorithm is similar to the deterministic one presented in~\cite{Sosik2007a}, although with an explicit stack instead of a recursive definition, and the further difference that their algorithm was able to work for unbounded-depth system. The actual algorithm implemented to answer the query is presented in Algorithm~\ref{alg:inner}.

\begin{algorithm}
  \SetKw{KwGuess}{guess}
  \SetKw{KwPop}{pop}
  \SetKw{KwPush}{push$_S$}
  $S \leftarrow \varnothing$ \; \label{alg2:init-stack}
  \For{all internal membrane $[w]_k^\alpha$ in the initial configuration}{  \label{alg2:push-initial-conf}
    \KwPush $(w, k, \alpha, 0)$ \;  \label{alg2:push-membrane-conf}
  }  \label{alg2:push-initial-conf-end}
  \While{$S$ is not empty}{  \label{alg2:main-loop}
    $(w, k, \mathsf{charge}, t_{\textrm{push}}) \leftarrow $ \KwPop $S$ \;  \label{alg2:pop-configuration}
    \For{$t' \leftarrow t_{\textrm{push}}$ \KwTo $t$}{  \label{alg2:for-the-remaining-time-steps}
      \For{$r = \pdivide{k}{\alpha}{a}{\beta}{b}{\gamma}{c}$ applicable}{  \label{alg2:apply-division}
        $m \leftarrow $ \KwGuess$(0,1)$\; \label{alg2:guess-division}
        \If{$m = 1$}{ \label{alg2:division-is-applied}
          mark a copy of $a$ for removal, a copy of $b$ for addition to $w$\; \label{alg2:mark-division}
          mark $\mathsf{charge}$ to be changed to $\beta$\; \label{alg2:mark-division-charge}
        } \label{alg2:division-is-applied-end}
      } \label{alg2:apply-division-end}
      \For{$r = \psendin{k}{\alpha}{a}{\beta}{b}$ applicable}{  \label{alg2:apply-send-in}
        $m \leftarrow $ \KwGuess$(0,1)$\;  \label{alg2:guess-send-in}
        \If{$m = 1$}{  \label{alg2:send-in-us-applied}
          $\iTable(r, t') \leftarrow \iTable(r, t') - 1$\;  \label{alg2:update-table-send-in}
          mark a copy of $b$ for addition to $w$\;  \label{alg2:mark-send-in}
          mark $\mathsf{charge}$ to be changed to $\beta$\;  \label{alg2:mark-send-in-charge}
        }  \label{alg2:send-in-is-applied-end}
      }  \label{alg2:apply-send-in-end}
      \For{$r = \psendout{k}{\alpha}{a}{\beta}{b}$ applicable}{  \label{alg2:apply-send-out}
        $m \leftarrow $ \KwGuess$(0,1)$\;  \label{alg2:guess-send-out}
        \If{$m = 1$}{  \label{alg2:send-out-is-applied}
          $\iTable(r, t') \leftarrow \iTable(r, t') - 1$\;  \label{alg2:update-table-send-out}
          mark a copy of $a$ for removal from $w$\;  \label{alg2:mark-send-out}
          mark $\mathsf{charge}$ to be changed to $\beta$\;  \label{alg2:mark-send-out-charge}
        }  \label{alg2:send-out-is-applied-end}
      }  \label{alg2:apply-send-out-end}
      \For{$r = \pevolve{k}{\alpha}{a}{u}$ applicable}{  \label{alg2:apply-evolve}
        $m \leftarrow $ \KwGuess$(0,|w|_a)$\;  \label{alg2:guess-evolution}
        mark $m$ copies of $a$ for removal, $m$ copies of $u$ for addition to $w$\;  \label{alg2:mark-evolution}
      }  \label{alg2:apply-evolve-end}
      apply marked modifications to $w$ and $\mathsf{charge}$\;  \label{alg2:modify-state}
      \If{rule application was not maximally parallel}{ \label{alg2:not-maximally-parallel}
        reject\; \label{alg2:reject-no-maximaly parallel}
      } \label{alg2:not-maximally-parallel-end}
      \textbf{if} division was applied, \KwPush $(w - \{b\} + \{c\}, k, \gamma, t')$\;  \label{alg2:push-if-division}
    }  \label{alg2:for-the-remaining-time-steps-end}
    \If{the current membrane has further applicable rules}{ \label{alg2:check-if-halted}
      reject\; \label{alg2:reject-not-halted}
    }
  }  \label{alg2:main-loop-end}
  \eIf{each entry of $\iTable$ is $0$}{  \label{alg2:check-guess}
    accept\;  \label{alg2:acceptance}
  }{
    reject\;  \label{alg2:rejection}
  }  \label{alg2:check-guess-end}
  \caption{\label{alg:inner} The nondeterministic polynomial space algorithm simulating the inner membranes of $\Pi_x$.}
\end{algorithm}

Lines~\ref{alg2:init-stack}--\ref{alg2:push-initial-conf-end} perform the initial set-up, where a new stack $S$ is filled with the configuration of all internal membranes at the initial time step, i.e., $t = 0$. In particular, for each membrane the multiset of objects contained, label, charge, and time step of the simulation are all pushed as an single record into $S$.

In the main loop of lines~\ref{alg2:main-loop}--\ref{alg2:main-loop-end} the simulation of all internal membranes is performed one at a time. This loop is executed until the stack of membranes to be simulated is not empty, which might require an exponential amount of time.

Once a new membrane to be simulated \emph{starting at time $t_{\mathsf{push}}$} has been extracted (line~\ref{alg2:pop-configuration}) the simulation of the membrane proceeds up to time step $t$, which is given in input as part the query (loop of lines~\ref{alg2:for-the-remaining-time-steps}--\ref{alg2:for-the-remaining-time-steps-end}) and represents the time at which the simulation of the outermost membrane has suspended in order to perform the query.

In lines~\ref{alg2:apply-division}--\ref{alg2:apply-division-end}, for each applicable division rule, i.e., the correct object and charge are present and the membrane has not already been used by a blocking rule in this time step, a nondeterministic choice is performed (line~\ref{alg2:guess-division}) to decide if the rule is actually applied. If so (lines~\ref{alg2:apply-division}--\ref{alg2:apply-division-end}), then the modifications described by the first half of the right-hand-side of the rule are performed, while the other membrane resulting from the division will be pushed on the stack $S$ at the end of the simulation of the current time step (line~\ref{alg2:push-if-division}). This cannot be performed earlier since the rewriting rules are applied, by the semantics of rule application in P~systems, before the division actually takes place.

The simulation of both send-in and send-out rules (lines~\ref{alg2:apply-send-in}--\ref{alg2:apply-send-in-end} and lines~\ref{alg2:apply-send-out}--\ref{alg2:apply-send-out-end}, respectively) is performed similarly. Since we are working in a situation of non-confluence, even if a rule is applicable, in order  to actually decide whether to apply it, a nondeterministic guess is performed (line~\ref{alg2:guess-send-in} and line~\ref{alg2:guess-send-out}, respectively). In both cases the modifications to be performed to the membrane configuration are scheduled for later execution (lines~\ref{alg2:mark-send-in}--\ref{alg2:mark-send-in-charge} and lines~\ref{alg2:mark-send-out}--\ref{alg2:mark-send-out-charge}, respectively). Since send-in and send-out are communication rules between the outermost membrane and the internal membranes, each time one of them is applied the value of $\iTable(r, t')$ is decremented. If, at the end of the simulation, the number of guessed applications and the real number of applications of the communication rules coincides, all entries $\iTable(r, t')$ should be $0$ (at line~\ref{alg2:update-table-send-in} and line~\ref{alg2:update-table-send-out}, respectively).

The application of evolution rules (lines~\ref{alg2:apply-evolve}--\ref{alg2:apply-evolve-end}), their effect being limited to the internal state of the membrane, is simpler. As usual, which rules are actually applied is determined by a nondeterministic choice (line~\ref{alg2:guess-evolution}).

Once all rule applications have been decided, the actual modifications to the state of the membrane are applied (line~\ref{alg2:modify-state}) and, if the rule application was not maximally parallel then the computation rejects (lines~\ref{alg2:not-maximally-parallel}--\ref{alg2:not-maximally-parallel-end}). This can be verified by checking if there still exist objects inside the membrane with applicable rules but no rule was applied to them, or if $\uTable(a, t')$ is positive for some $a \in \Gamma$ with an applicable send-in rule to the currently simulated membrane. Since $\uTable(a, t')$ indicates the number of objects that were available for the application of send-in from the outermost membrane but no internal membrane was available, such an inconsistency would denote that the simulation of the internal membranes had no correspondence to the already performed simulation of the outermost membrane.

If a division rule was applied, then the configuration of the second membrane resulting from division is pushed to the stack $S$ (line~\ref{alg2:push-if-division}). Here, an instance of the object $b$ has been replaced by an instance of object $c$ and the charge has been changed from $\beta$ to $\gamma$ to obtain from the current membrane a copy corresponding to the other one obtained by division.

Before proceeding with the simulation of another membrane, we check that after $t$ steps the computation in this membrane has actually halted (lines~\ref{alg2:check-if-halted}--\ref{alg2:reject-not-halted}). Otherwise the current computation must reject (line~\ref{alg2:reject-not-halted}).

After the simulation of all internal membranes is finished, i.e., the stack was emptied, a check on the entries of $\iTable$ is performed. If all and every communication rule application guessed during the simulation of the outermost membrane was actually executed then all entries of $\iTable$ should be $0$. A positive (resp., negative) value for $\iTable(r, t)$ denotes that less (resp., more) applications of rule $r$ at time $t$ were performed than the number that was guessed.

If at least one accepting computation of the machine simulating the oracle query exists then the answer to the query is positive. Furthermore, if there is a way to ``glue'' the simulation of the outermost membrane and of the internal membranes, then the result produced by Algorithm~\ref{alg:outermost} was correct. Combining this simulation with the inverse simulation presented in~\cite{Leporati2016a}, we can then state the main result of the paper.

\begin{theorem}
  \label{thm:pspace}
  $\PSPACE = \NPMC_{\pAM[depth\mbox{-}1,-d,-ne]}^{[\star]}$. \qed
\end{theorem}

As long as no dissolution is allowed, the property of being elementary is a static one and, if no non-elementary division is present, the simulation of the outermost membrane can be extended to include all non-elementary membranes, allowing us to state the following result.

\begin{corollary}
    $\PSPACE = \NPMC_{\pAM[-d,-ne]}^{[\star]}$. \qed
\end{corollary}

\section{Conclusions}
\label{sec:conclusions}

We have shown that, differently from confluent P~systems, monodirectionality and a restriction on the depth of the system to $1$ (or, equivalently, the absence of both dissolution and non-elementary division) do not prevent non-confluent P~systems from reaching $\PSPACE$ in polynomial time. It remains open to establish if this upper bound can be extended to membrane structures of higher (non-constant) depth where non-elementary division is allowed. Since both monodirectionality and nesting depth have a huge influence in the computational power of confluent systems, it would be worthwhile to understand why they do not provide an analogous increase to non-confluent systems. These features are usually employed by algorithms designed for confluent P~systems to simulate the power of nondeterminism, so the question is: are they always useless when non-confluence is already present?

\bibliographystyle{splncs03}
\bibliography{Bibliography}

\end{document}